\documentclass[journal]{IEEEtran}

\renewcommand{\baselinestretch}{0.98}
\ifCLASSINFOpdf
\else
   \usepackage[dvips]{graphicx}
\fi
\usepackage{url}

\usepackage{graphicx}
\usepackage[utf8]{inputenc}
\usepackage{amsmath,amsthm,amssymb,amsfonts,bm}
\usepackage{algorithmic}
\usepackage{graphicx}
\usepackage{tikz}
\usepackage{hyperref}
\usepackage{cite}
\usepackage{balance}

\newcommand{\Rev}[1]{\textcolor{black}{#1}}

\newtheorem{theorem}{Theorem}
\begin{document}

\title{An Uncertainty-Aware Performance Measure for Multi-Object Tracking}

\author{Juliano Pinto, Yuxuan Xia, \textit{Student Member, IEEE}, Lennart Svensson, \textit{Senior Member, IEEE}, and Henk Wymeersch, \textit{Senior Member, IEEE}
\thanks{This work was supported, in part, by a grant from the Chalmers AI Research Centre Consortium (CHAIR).}
\thanks{The authors are with the Department of Electrical Engineering, Chalmers University of Technology, Sweden.  (e-mail: juliano@chalmers.se).}}

\maketitle

\begin{abstract}
Evaluating the performance of multi-object tracking (MOT) methods is not straightforward, and existing performance measures fail to consider all the available uncertainty information in the MOT context. 
This can lead practitioners to select models which produce uncertainty estimates of lower quality, negatively impacting any downstream systems that rely on them. Additionally, most MOT performance measures have hyperparameters, which makes comparisons of different trackers less straightforward. 
We propose the use of the negative log-likelihood (NLL) of the \Rev{multi-object posterior} given the set of ground-truth objects as a performance measure. This measure takes into account all available uncertainty information in a sound mathematical manner without hyperparameters. 
We provide efficient algorithms for approximating the computation of the NLL for several common MOT algorithms, show that in some cases it decomposes and approximates the widely-used GOSPA metric, and provide several illustrative examples highlighting the advantages of the NLL in comparison to other MOT performance measures.
\end{abstract}


\IEEEpeerreviewmaketitle

\vspace{-4mm}

\section{Introduction}
\label{sec:introduction}

Multi-object tracking (MOT) is the task of tracking an unknown number of objects through time using noisy measurements, with important applications in various areas \cite{bar2004estimation,challa2011fundamentals,meyer2018message,streit2021analytic,mahler2007statistical}.
The main challenge for MOT is the unknown correspondence between objects and measurements, which makes it necessary for the algorithms to infer such information 
\cite{bar1990multitarget}.
Additionally, in many applications it is important that tracking systems provide accurate uncertainty estimates of their outputs, so that decision-making systems can take robust actions \cite{uncertainty_in_machine_learning}.
%
Evaluating the quality of 
MOT methods is also challenging, 
due to the lack of knowledge about the correct correspondence between the ground-truth set of object states and the tracker's state estimates and the associated uncertainties \cite{evaluating_mot_performance}. 
%

\Rev{Several performance measures have been proposed for MOT, including the Hausdorff metric \cite{wasserstein_hausdorff_metrics}, the Wasserstein metric \cite{wasserstein_hausdorff_metrics}, multi-object tracking accuracy and precision (MOTA and MOTP) \cite{clearmot}, higher-order tracking accuracy (HOTA) \cite{luiten2021hota}, optimal subpattern assignment (OSPA) \cite{schuhmacher2008consistent}, and generalized optimal subpattern assignment (GOSPA) \cite{rahmathullah2017generalized}. MOTA and MOTP are prevalent in computer vision tasks, such as multiple people tracking \cite{mot20}, while OSPA and GOSPA are the most widely applied for general MOT problems}. These performance measures cope with the unknown associations between estimates and objects by relying on minimum-cost associations based on a user-defined distance measure. 
None of these methods can assess the quality of the uncertainty estimates, and existing measures such as normalized estimation error squared \cite{nees} are not easily applicable to MOT due to the unknown data associations.

To address this, there have been efforts to extend the MOT performance measures to incorporate uncertainty information, but they only evaluate some of the available uncertainty \cite{nagappa2011incorporating,he2013track}, and/or rely on proxies for ground-truth state uncertainties which are not applicable to general MOT applications \cite{he2013track}. To the best of our knowledge, no general, mathematically sound incorporation of all the uncertainties in MOT has been proposed in a performance measure. 
\Rev{Additionally, most MOT performance measures have hyperparameters (e.g., distance metric and thresholds). 
Such choices are often non-trivial, difficult to generalize to new contexts, and make comparisons between different trackers less straightforward. Moreover, MOT methods themselves may have hyperparameters, optimized for a certain metric, leading to a chicken-and-egg-problem.}

In this letter, we introduce a new MOT performance measure that incorporates all uncertainties in a sound mathematical manner, while at the same time having zero hyperparameters: the negative log-likelihood (NLL) of the model (i.e., the MOT method posterior), given the ground-truth objects.
NLL has been widely applied in 
statistics \cite{severini2000likelihood} and optimization \cite{myung2003tutorial}, but not yet in MOT. 
We provide the following contributions. First, we propose the use of the NLL as a new MOT performance measure, and provide efficient algorithms for computing/approximating it for a number of common MOT algorithms based on random finite sets (RFSs): (C)PHD \cite{mahler2003multitarget,mahler2007phd}, PMBM \cite{garcia2018poisson}, PMB \cite{williams2015marginal}, MBM \cite{garcia2019gaussian}, and MBM$_{01}$ \cite{vo2013labeled}. Second, we show that for the PMB MOT family, the NLL decomposes into separate terms that provide additional transparency into the performance, and that under certain assumptions, NLL is closely related to GOSPA. Lastly, we provide illustrative examples to highlight the advantages of NLL. 

\subsubsection*{Notations}
Scalars and vectors are denoted by lowercase or uppercase letters with no special typesetting $x$, matrices by uppercase boldface letters $\mathbf X$, and sets by uppercase blackboard-bold letters $\mathbb X$. In addition, we define $\mathbb N_a = \{i \in \mathbb N~|~i\leq a\}, a\in\mathbb N$. 

\vspace{-4mm}
\section{NLL as a performance measure}
In this section, we provide the definition of NLL, along with examples on how to compute it efficiently for different families of MOT densities. Given a ground-truth set of object states $\mathbb Y=\{y_1, \cdots, y_{|\mathbb Y|}\}$ and a posterior density of the tracked objects $f_{\mathsf{M}}(\cdot)$ from method $\mathsf{M}$ \Rev{(i.e., the multi-object posterior, which describes the distribution of the set of object states)}, the NLL is defined as 
\begin{align}
    \label{eq:nll_definition}
    \mathrm{NLL}(\mathbb Y,f_{\mathsf{M}})=-\log f_{\mathsf{M}}(\mathbb Y).
\end{align}
We can then rank different algorithms (say $\mathsf{M}_1$ and $\mathsf{M}_2$) by comparing $\mathrm{NLL}(\mathbb Y,f_{\mathsf{M}_1})$ and $\mathrm{NLL}(\mathbb Y,f_{\mathsf{M}_2})$
for the same set $\mathbb Y$, or by computing an expectation with respect to different trials. 
\Rev{This type of comparison naturally incorporates the uncertainty information estimates by the trackers, since to score well algorithms must have most of the mass of their posterior in regions where it is likely that the objects in $\mathbb Y$ will be, without being overly confident (see Section \ref{sec:examples})}. We now provide examples of several important MOT filters and show how the NLL can be computed. 

\subsubsection{(C)PHD Filters}
(C)PHD filters \cite{mahler2003multitarget,mahler2007phd} have been widely used due to their low complexity, simple implementation, and relatively good performance. The multi-object posterior for a CPHD filter $f_\text{CPHD}(\mathbb{X})$ takes the form
\begin{equation}
    \label{eq:cphd_posterior}
    f_\text{CPHD}(\mathbb X) = |\mathbb{X}|!p(|\mathbb{X}|)\prod_{x\in\mathbb{X}}s(x),
\end{equation}
where $p(|\mathbb{X}|)$ is the cardinality distribution (a Poisson distribution for the PHD filter) of $\mathbb{X}$, and $s(x)$ is the single-object state density.  Computing $\mathrm{NLL}(\mathbb Y,f_{\text{CPHD}})$ according to \eqref{eq:nll_definition} therefore yields the expression
\begin{align}
    \mathrm{NLL}(\mathbb Y,f_{\text{CPHD}}) = -\log(|\mathbb Y|!)-\log p(|\mathbb{Y}|) -\sum_{y\in\mathbb{Y}} \log s(y),\notag
\end{align}
which can be computed with complexity $\mathcal{O}(|\mathbb{Y}|)$. 

\subsubsection{PMBM Filters}
\label{subsec:pmbm_likelihood}
Poisson multi-Bernoulli mixture (PMBM) filters \cite{williams2015marginal,garcia2018poisson} are the optimal solution to MOT with standard multi-object dynamic and measurement models with Poisson birth \cite[Chap. 13]{mahler2007statistical}. The multi-object posterior for a PMBM filter is defined as:
\begin{align}
    \label{eq:pmbm_dist}
    f_\text{PMBM}(\mathbb X) = \sum_{\mathbb{X}^{\text{U}} \uplus \mathbb{X}^{\text{D}}=\mathbb X}f_\text{PPP}(\mathbb{X}^{\text{U}})f_\text{MBM}(\mathbb{X}^{\text{D}})
    \\
    \label{eq:ppp_dist}
    f_\text{PPP}(\mathbb X)=\exp\Big(-\int\lambda(x')\mathrm{d}x'\Big)\prod_{x\in\mathbb X}\lambda(x)
    \\
    \label{eq:mbm_dist}
    f_\text{MBM}(\mathbb X)=\sum_{h=1}^H w_h\sum_{\uplus_{j=1}^{m}\mathbb X_j=\mathbb X}\prod_{k=1}^{m}f_k^h(\mathbb X_k)
\end{align}
where $\lambda(\cdot)$ is the intensity function of the Poisson point process (PPP), $w_h$ are the weights of each of the $H$ MB components of the MBM ($\sum_h w_h = 1$), $m$ is the number of Bernoulli components in each of the MB components of the MBM (set as identical for each mixture component without loss of generality), and $f_k^h(\mathbb X_k)$ is the $k$-th Bernoulli density of the $h$-th hypothesis, with
\begin{align}
    \label{eq:bernoulli_dist}
    f_k^h(\mathbb X_k)=\begin{cases}
    1-r_k^h~, & \text{if }\mathbb X_k = \emptyset \\
    r_k^h p_k^h(x)~, & \text{if }\mathbb X_k = \{ x\} \\
    0 & \text{otherwise}.
    \end{cases}
\end{align}
where $r_k^i$ is the existence probability, and $p_k^i(\cdot)$ is the single-object density.

When evaluating $f_\text{PMBM}(\mathbb Y)$,  \eqref{eq:pmbm_dist} can be interpreted as summing the likelihoods of all the possible assignments between the elements of $\mathbb Y$ and either the PPP component or one of the Bernoulli components of the PMB density, for each hypothesis $h$.
Since the number of such possible assignments grows super-exponentially in $\mathbb{X}$, computing the NLL has a complexity that also grows at the same rate in $\mathbb Y$. However, among all these assignments, generally only a few contribute significantly to the overall sum, and the likelihood can be approximated by neglecting all other terms (e.g., when the ground-truth objects in $\mathbb Y$ are reasonably well-separated and so are the $p_k^h$ for each $h$). 

To find such terms for each hypothesis $h$, we solve an optimal assignment problem for $\bm{A}^{h}\in \{0,1\}^{(m+|\mathbb Y|) \times |\mathbb Y|}$ \cite[Chap. 7]{bar1990multitarget}:
\begin{subequations}
\label{eq:assignment_problem}
\begin{align}
    \min_{\bm{A}^h} \quad & \sum_k\sum_{l} C^{h}_{k, l}A^{h}_{k,l}
    \\
    \label{eq:assignment_problem_constraints}
    \text{s.t.} \quad & \sum_{k=1}^{m+|\mathbb Y|} A^{h}_{k,l}=1,\, \sum_{l=1}^{|\mathbb Y|} A^{h}_{k,l}\leq 1,
\end{align}
\end{subequations}
where $\bm{C}^h$ is a cost matrix defined as
\begin{equation}
    \label{eq:cost_matrix}
    C^{h}_{k, l} = \begin{cases} 
    -\log\left(\frac{p^{h}_k(y_l)}{1-r^{h}_k}r^{h}_k\right), &\text{if }k\leq m
    \\
    -\log\lambda(y_l),& \text{if } k=l+m
    \\
    \infty,& \text{otherwise,}
    \end{cases}
\end{equation} 
and $\mathbf A^h$ is the assignment matrix between ground-truth objects and the components of the PMBM. If $[\mathbf A^h]_{i,j}=1$, then $y_j$ is assigned to the $i$-th component of the PMBM, where $1\leq i\leq m$ corresponds to the $m$ Bernoulli components in hypothesis $h$, and all $i>m$ to the PPP component.

Murty's algorithm \cite{assignment_problem_complexity} allows for efficient computation of the $Q$-lowest cost associations $\bm{A}^{h,*}_1, \cdots, \bm{A}^{h,*}_Q$ to this assignment problem. We find that
\begin{align}
    \label{eq:approximated_nll_pmbm}
    &\mathrm{NLL}(\mathbb Y,f_{\text{PMBM}}) \approx  \int\lambda(y')\mathrm{d}y'\\
    &
    -\log \Big( \sum_{h=1}^H\sum_{q=1}^Q w_h \prod_{y\in\mathbb{ Y}^{\text{U}}(\bm{A}_q^{h,*})}\lambda(y)\prod_{k=1}^{m}f_k^h(\mathbb{Y}_k(\bm{A}_q^{h,*})) \Big) \notag
\end{align}
where $\mathbb Y_k(\mathbf A_q^{h,*})=\{y_j\in\mathbb Y ~|~ [\mathbf A_q^{h,*}]_{k,j}=1\}$, $\mathbb Y^\text{U}(\mathbf A_q^{h,*})=\mathbb Y \setminus \cup_{i=1}^m \mathbb Y_i(\mathbf A_q^{h,*})$. 
The worst-case time complexity of this approximation scales as $\mathcal{O}(HQ (m+|\mathbb{Y}|)^3)$ \cite{assignment_problem_complexity}. 
\subsubsection{Other Common MOT Filters}
The same assignment problem defined in \eqref{eq:assignment_problem} can be used to efficiently compute other special cases of the PMBM density. For instance, the PMB density is a special case of the PMBM with a single MB component ($H=1$) \cite{mahler2007statistical}, and the MBM density is a PMBM where $\lambda(x)=0,\,\forall x$ \cite{mahler2007statistical}. Additionally, the MBM$_{01}$ density is a special case of the MBM density when existence probabilities $p_k^i$ of all the Bernoulli components are set to either 0 or 1 \cite{garcia2018poisson}. For all these densities, the multi-object posterior has the same form as \eqref{eq:pmbm_dist}, and can therefore be efficiently approximated using \eqref{eq:approximated_nll_pmbm}. \Rev{Finally, any MOT filter that produces a set of predictions with state uncertainties and/or existence probabilities (such as some deep-learning-based methods, e.g.,  \cite{deep_learning_detection_with_uncertainty})
can be seen as having a PMBM density posterior with $H=1$ and $\lambda(x)=0~\forall x$, and therefore its NLL can also be approximated by \eqref{eq:approximated_nll_pmbm}}. 

\vspace{-3mm}

\section{Decomposition of the NLL for PMB densities}
An attractive property for a performance measure in MOT is that it decomposes into meaningful terms, which can then be individually analyzed for providing additional insights into the types of errors made by the algorithms. In this section we show that the NLL for PMB densities decomposes into three separate terms that depend on the ability of the posterior to explain matched objects, missed objects, and false detections, while taking into account all the uncertainties in the posterior. In addition, we provide a connection between the decomposed form of the NLL and the GOSPA metric, showing the under certain conditions the latter is a special case of the former, up to additive offsets.

\vspace{-3mm}
\subsection{Decomposing the NLL}
As a special case of \eqref{eq:approximated_nll_pmbm} with $H=1$ and $Q=1$, the negative log-likelihood (NLL) of a PMB density can be approximated as:
\begin{align}
    \label{eq:pmb_nll}
    &\mathrm{NLL}(\mathbb Y,f_{\text{PMB}}) \approx  \\
    & \min_{\mathbf{A}\in\mathbb A}-\sum_{k=1}^m \log f_k(\mathbb Y_k(\mathbf{A}))+\int\lambda(y')\mathrm{d}y' -\sum_{y\in\mathbb Y^{\text{U}}(\mathbf{A})}\log\lambda(y)~, \notag
\end{align}
where we drop the dependency on $h$ and $q$, and where $\mathbb A$ is the set of all matrices in $\mathbb \{0, 1\}^{(m+|\mathbb Y|)\times|\mathbb Y|}$ that satisfy \eqref{eq:assignment_problem_constraints}.
We can also express the optimization in \eqref{eq:pmb_nll} in terms of assignment sets $\gamma$, i.e., the set of matched indices $(i,j)$, where $\gamma(\mathbf A)=\left\{(i,j)\in \mathbb N_m\times\mathbb N_{|\mathbb Y|}~|~ [\mathbf A]_{i,j}=1\right\}$. 
This change, together with \eqref{eq:bernoulli_dist}, yields
\begin{align}
    \label{eq:pmb_nll_decomposition}
    &\textrm{NLL}(\mathbb Y,f_\text{PMB})\approx\min_{\gamma\in\Gamma}
    \underbrace{-\sum_{(i,j)\in\gamma}\log \big(r_ip_i(y_j)\big)}_\text{Localization}
    \\
    &\underbrace{-\sum_{i\in\mathbb F(\gamma)}\log(1-r_i)}_\text{False detections}~
    \underbrace{+\int \lambda(y')\textrm{d}y'-\sum_{j\in\mathbb M(\gamma)}\log\lambda(y_j)}_\text{Missed objects}~,\notag
\end{align}
where $\Gamma$ is the set of all possible assignment sets, $\mathbb F(\gamma)=\left\{ i\in\mathbb N_m ~|~ \nexists~j:(i,j)\in\gamma \right\}$ is the set of indices of the Bernoullis not matched to any ground-truth, and $\mathbb M(\gamma)=\left\{j\in\mathbb N_{|\mathbb Y|}~|~\nexists~i: (i,j)\in\gamma\right\}$ is the set of indices of ground-truths not matched to any Bernoulli component. 
This decomposition separates the NLL score of a PMB density into three parts: the \textit{localization} part, which measures how well the Bernoulli components explain the matched ground-truth, the \textit{false detections} part, which measures the the existence probabilities for unmatched Bernoullis, and the \textit{missed objects} part, which measures how well the PMB's PPP component explains the objects not matched to any Bernoulli component. Such decomposition allows practitioners to identify which types of mistakes the algorithms being evaluated are committing.
\vspace{-3mm}

\subsection{Connection to GOSPA}
We proceed to show that  \eqref{eq:pmb_nll_decomposition} can be related to the GOSPA metric \cite{rahmathullah2017generalized}. Recall that  the GOSPA metric ($p=1$, $\alpha = 2$) between two sets $\mathbb X$ and $\mathbb Y$ is
\begin{equation*} 
    \mathrm{GOSPA}_c(\mathbb X, \mathbb Y) = \min_{\gamma\in\Gamma} \sum_{(i,j)\in\gamma}D(x_i, y_j) + \frac{c}{2}(|\mathbb X|+|\mathbb Y|-2|\gamma|),
\end{equation*}
where $D(\cdot,\cdot)$ is a metric on the state-space and $c$ is a cut-off distance. Observe that similarly to NLL, GOSPA decomposes into individual terms for the quality of the localized objects and penalties for false detections and missed objects, but does not account for existence probabilities or object densities. This similarity can be made explicit in the following theorem.  
\begin{theorem}
If (i) $\lambda(x)={\bar\lambda}/{V} \le 1$, inside a field-of-view with volume $V$ (0 elsewhere), (ii)  $r_i=\rho$, $\forall i$, (iii) $p_i(x) =  \exp(-D(x_i,x))/k$, $\forall i$, where $D$ is a translation-invariant metric,\footnote{Translational invariance is needed to allow $k$ to be independent of $x_i$.} $k$ is the normalization constant, and (iv) $\log(1-\rho)=\log(\bar\lambda/V)\doteq -c/2$, then 
\begin{align} 
        &\mathrm{NLL}(\mathbb Y,f_\text{PMB}) = V(1-\rho)
        \\
        &+ \min_\gamma\sum_{(i,j)\in\gamma}\Big(D(x_i, y_j)+\log \frac{k}{\rho}\Big) + \frac{c}{2} (|\mathbb X| +|\mathbb Y|- 2|\gamma|) \notag
\end{align}
where $\mathbb{X}=\{x_1,\ldots, x_m\}$. 
\end{theorem}
\begin{proof}
See Appendix \ref{sec:ProofThm1}. 
\end{proof}
The necessary conditions can be interpreted as (i) the Poisson intensity function takes a constant value over the field-of-view with volume $V$; (ii) all the Bernoulli existence probabilities are identical; (iii) the state distribution for all the Bernoulli components of the PMB are set to the same functional form; (iv) there is a relation between the existence probabilities and the Poisson intensity. This result shows that when such conditions hold, GOSPA is a special case of the NLL performance measure up to certain additive constants. 

\vspace{-3mm}
\section{Illustrative examples} \label{sec:examples}
In this section, we study the properties of the NLL performance measure in two illustrative examples and compare it to \Rev{three popular MOT performance measures: MOTA \cite{clearmot}, MOTP \cite{clearmot}, and GOSPA \cite{rahmathullah2017generalized}. We use $2.0$ as the cutoff distance for establishing matches, Euclidean distance to compute errors, and $\alpha=2$, $p=1$ for the GOSPA hyperparameters.}
\subsubsection*{Example 1}
Consider two trackers, $\mathsf{M}_1$ and $\mathsf{M}_2$ in a MOT problem where the state is 2-dimensional position, and $\mathbb Y=\{(2, 5), (6, 3)\}$. \Rev{Both trackers have a posterior in the form of a two-component MB with Gaussian state densities (which can be seen as PMBM filters), illustrated in Fig.\,\ref{fig:example_1} along with the corresponding existence probability for each component}.
$\mathsf{M}_1$'s predictions are better than $\mathsf{M}_2$'s: its Bernoulli component closer to A is more certain of the object's existence, and the state density for the component closer to B explains the ground-truth position significantly better than the component from $\mathsf{M}_2$. \Rev{In order to compute MOTA, MOTP, and GOSPA scores for each tracker}, one would first need to extract ``hard'' estimates from these posteriors. A common practice is to use the means of the Bernoulli state densities \cite{garcia2018poisson}, resulting in prediction sets $\mathbb X_{\mathsf{M}_1}=\{(3, 5), (7,4)\}$ and $\mathbb X_{\mathsf{M}_2}=\{(1, 5), (5, 2)\}$. Since the ground-truth positions are equidistant to these estimates, all such performance measures rank $\mathsf{M}_1$ and $\mathsf{M}_2$ equally: \Rev{$\text{MOTA}=1$, $\text{MOTP}=\frac{1+\sqrt{2}}{2}$, and $\text{GOSPA}=1+\sqrt{2}$ for both $\mathsf{M}_1$ and $\mathsf{M}_2$}.
\Rev{In contrast, the NLL also uses the uncertainty information available in the posteriors of the trackers to compare $\mathsf{M}_1$ and $\mathsf{M}_2$.} Specifically, since the state densities are Gaussian, we have that
\Rev{\begin{equation}
    f_{\mathsf{M}_k}(\mathbb Y) \approx 
    \prod_{i=1}^2 r_{i,k} \frac{e^{-\frac{1}{2}(y_i-\mu_{i,k})^\top\Sigma_{i,k}^{-1}(y_i-\mu_{i,k})}}{2\pi\sqrt{|\Sigma_{i,k}|}} 
\end{equation}
because the other possible matches between Bernoullis and predictions have negligible contribution. Hence, applying the definition of the NLL from \eqref{eq:nll_definition}, we obtain}

\begin{align}
    &\textrm{NLL}(\mathbb Y, f_{\mathsf{M}_k})\approx
    \sum_{i=1}^2-\log r_{i,k} + 
    \frac{1}{2}\log|\Sigma_{i,k}| \notag
    \\
    \label{eq:nll_example1}
    &+\frac{1}{2}(y_i-\mu_{i,k})^\top \Sigma_{i,k}^{-1} (y_i-\mu_{i,k}) +\log(2\pi),
\end{align}
where $\mu_{i,k}$ and $\Sigma_{i,k}$ are the mean and covariance matrix of the $i$-th Bernoulli component of $\mathsf{M}_k$, and $y_i$ is the ground-truth matched to that component according to the optimal match.
\Rev{Note how these three terms incorporate all the uncertainty information available in the posterior:}  $(y_i-\mu_{i,k})^\top \Sigma_{i,k}^{-1} (y_i-\mu_{i,k})$ 
penalizes overconfident state densities, 
$-\log r_{i,k}$ accounts for the quality of the existence probabilities, while $\log|\Sigma_{i,k}|$  penalizes underconfident state densities. 
Evaluating \eqref{eq:nll_example1} for $\mathsf{M}_1$ and $\mathsf{M}_2$ results in $\textrm{NLL}(\mathbb Y, f_{\mathsf{M}_1})\approx4.6$ and $\textrm{NLL}(\mathbb Y, f_{\mathsf{M}_2})\approx16.3$, in line with the intuition that $\mathsf{M}_1$'s predictions are considerably better than $\mathsf{M}_2$'s (lower NLL scores are better, indicating higher likelihood).
\begin{figure}
    \centering
    \includegraphics[width=0.4\textwidth]{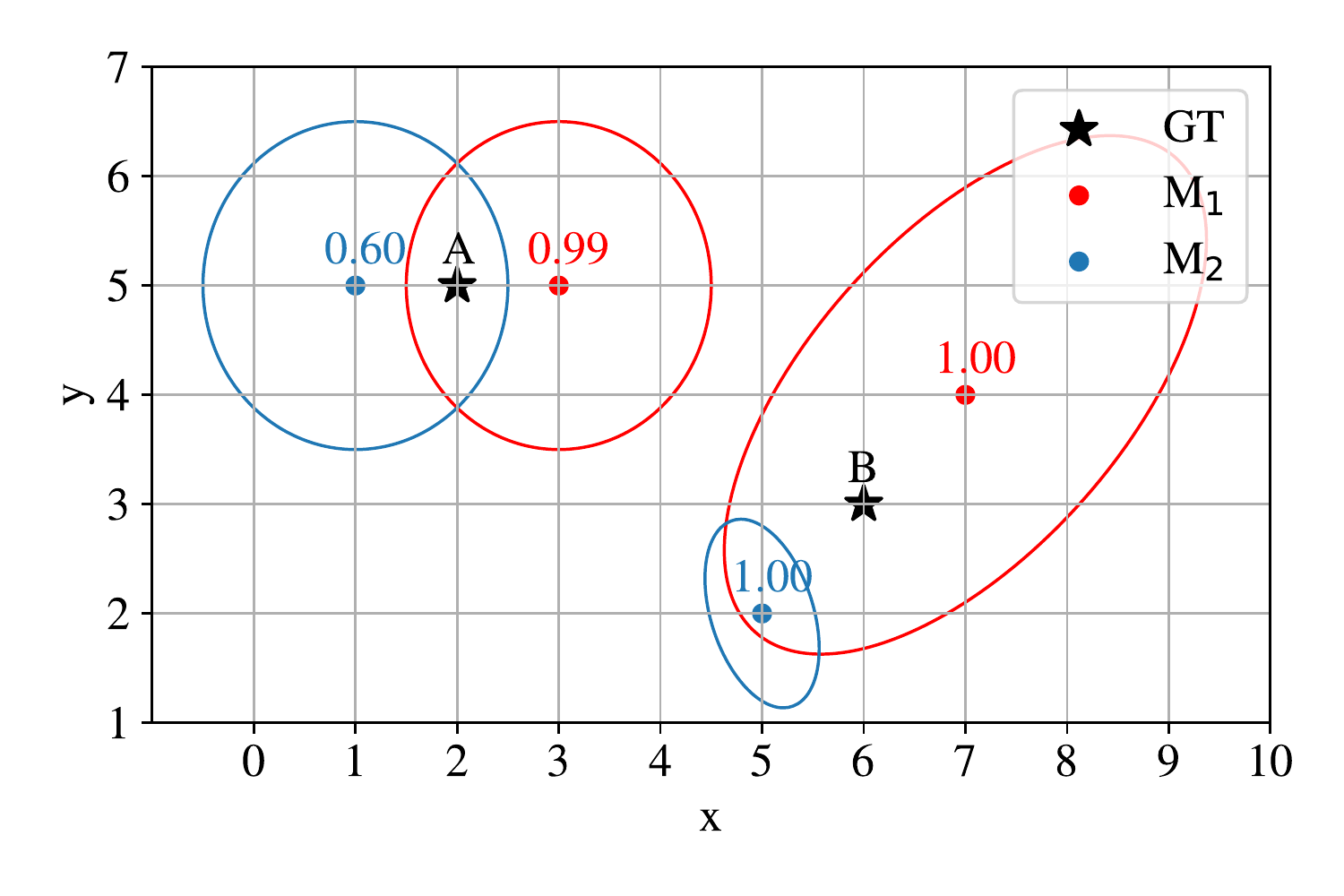}\vspace{-5mm}
    \caption{\emph{Example 1:} existence probabilities and 2-sigma regions for the Bernoulli components of $\mathsf{M}_1$ and $\mathsf{M}_2$, along with ground-truth object states as black stars.}
    \label{fig:example_1}
    \vspace{-4mm}
\end{figure}
\begin{figure}
    \centering
    \includegraphics[width=0.4\textwidth]{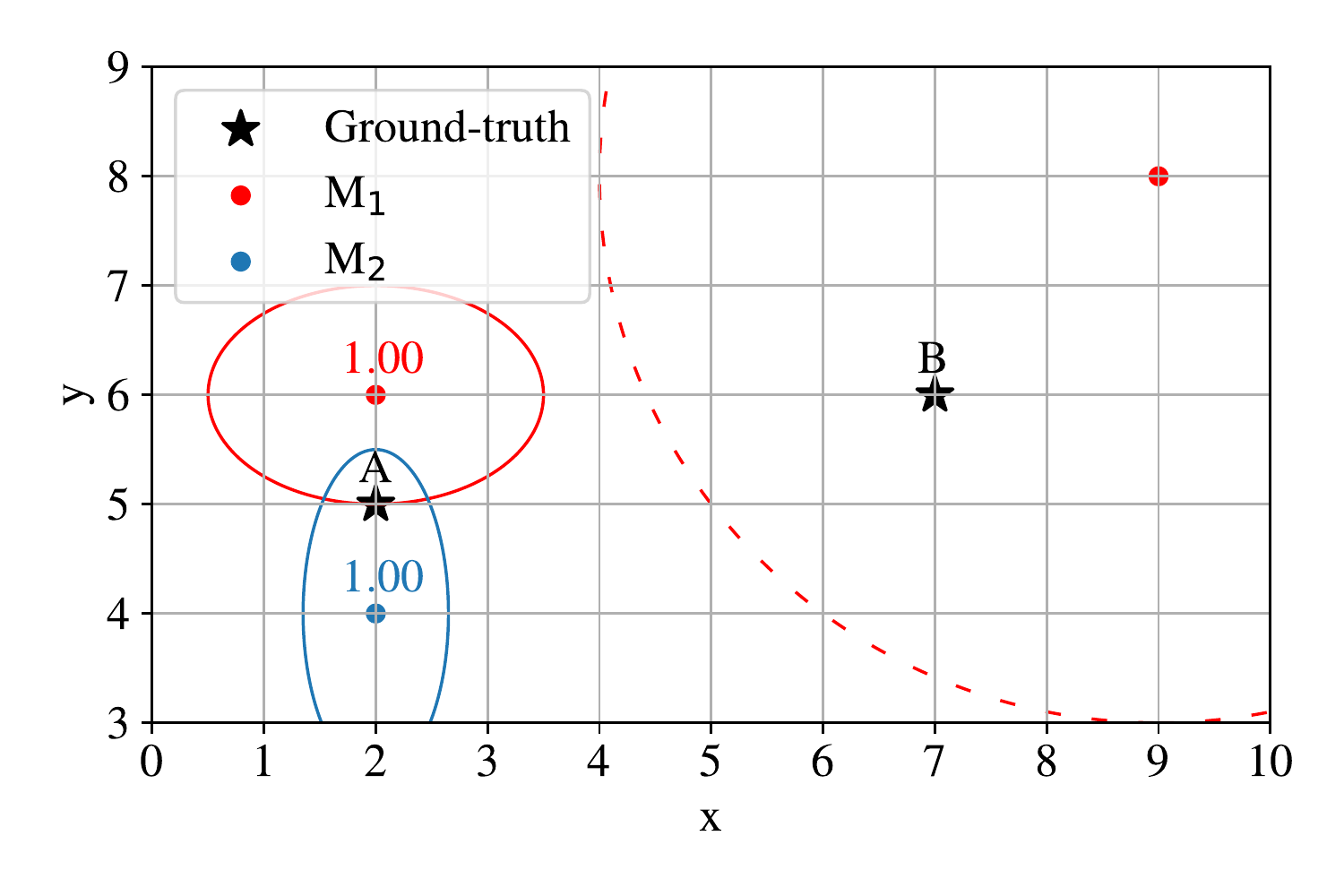}\vspace{-5mm}
    \caption{\emph{Example 2:} now $\mathsf{M}_1$ has an additional PPP component, with 2-sigma region of the intensity function $\lambda(\cdot)$ illustrated as a dashed ellipse. Same legend as Fig.\,\ref{fig:example_1}.  }
    \label{fig:example_2}
    \vspace{-4mm}
\end{figure}

\subsubsection*{Example 2}
In this example the state is also the 2-dimensional position of the objects, this time with $\mathbb Y=\{(2, 5), (7,6)\}$. \Rev{$\mathsf{M}_1$ is now a PMB filter (PMB density as its posterior), while $\mathsf{M}_2$ is a Bernoulli mixture with a single component, illustrated in Fig.\,\ref{fig:example_2}}. The PPP component of $\mathsf{M}_1$ has an intensity function $\lambda(\cdot)$ which is Gaussian, and its 2-sigma region is illustrated as a dashed ellipse.
From the figure, we see that the posterior from $\mathsf{M}_1$ is superior in this case too, being vastly better at explaining the ground-truth than $\mathsf{M}_2$. Although both trackers localize object A equally well, $\mathsf{M}_2$'s posterior is unable to explain any $\mathbb Y$ such that $|\mathbb Y|>1$ (likelihood 0), whereas the PPP component of $\mathsf{M}_1$ is able to explain any number of missed objects. Extracting estimates from these posteriors as in Example 1 results in $\mathbb X_{\mathsf{M}_1}=\{(2, 6)\}$ and $\mathbb X_{\mathsf{M}_2}=\{(2, 4)\}$. \Rev{Again, MOTA, MOTP, and GOSPA all rank $\mathsf{M}_1$ and $\mathsf{M}_2$ equally, since both estimates are equidistant to ground-truth object A, and both trackers miss the ground-truth object B: $\text{MOTA}=0.5$, $\text{MOTP}=1$, and $\text{GOSPA}=2$ for both trackers.} Therefore, all these performance measures completely miss the fact that according to $f_{\mathsf{M}_2}$ the realization $\mathbb Y$ is impossible.
\Rev{In contrast, the NLL performance measure takes into account the available uncertainty information, revealing the correct ranking of these trackers}. In this case, $\textrm{NLL}(\mathbb Y, f_{\mathsf{M}_2})=\infty$, since $(f_{\mathsf{M}_2}(\mathbb Y)=0)$, whereas
\begin{align}
    \textrm{NLL}(\mathbb Y, f_{\mathsf{M}_1})&\approx-\log \big(r_1 p_1(y_A)\big) + 1-\log\lambda(y_B) \approx 8.2, \notag
\end{align}
therefore showing that the NLL performance measure correctly ranks $\mathsf{M}_1$ to be the best tracker in this context.
\vspace{-3mm}

\section{Conclusion}
We proposed the use of the negative log-likelihood of a model as a MOT performance measure and showed that it incorporates all available uncertainties in the evaluation in a sound mathematical manner, without the need of any hyperparameters. 
We also provided efficient algorithms for approximating the NLL of common MOT methods, along with a special case for the PMB density where the NLL decomposes and is an approximate generalization of the GOSPA metric. Examples indicate that NLL is better at capturing the expected performance than conventional performance measures. 
\Rev{The adoption of NLL can hopefully provide a fairer way of comparing MOT methods and guide novel designs, especially for deep learning methods, which often fail to provide uncertainty estimates and thus have poor NLL.}

\appendices
\section{Proof of Theorem 1} \label{sec:ProofThm1}
Starting from \eqref{eq:pmb_nll_decomposition}:
\begin{align*}
        &\textrm{NLL}(\mathbb Y,f_\text{PMB})\approx\min_{\gamma\in\Gamma}-\sum_{(i,j)\in\gamma}\log r_ip_i(y_j)
        \\
        &-\sum_{i\in\mathbb F(\gamma)}\log(1-r_i)+\int \lambda(y')\textrm{d}y'-\sum_{j\in\mathbb M(\gamma)}\log\lambda(y_j)\notag
\end{align*}
and invoking assumptions (i)--(iv), we obtain
\begin{align*}
    &\stackrel{\text{(i)}}{=}\min_{\gamma\in\Gamma}-\sum_{(i,j)\in\gamma}\log r_ip_i(y_j)
\\
    &\qquad-\sum_{i\in\mathbb F(\gamma)}\log(1-r_i)+\bar\lambda -\left( \log\frac{\bar\lambda}{V}\right )(|\mathbb Y|-|\gamma|)
\\
    &\stackrel{\text{(ii)}}{=}\min_{\gamma\in\Gamma}-\sum_{(i,j)\in\gamma}\Big(\log p_i(y_j)+\log\rho\Big)
\\    
    &\qquad-\log(1-\rho)(|\mathbb X|-|\gamma|)+\bar\lambda -\left( \log\frac{\bar\lambda}{V}\right )(|\mathbb Y|-|\gamma|)
\\
    &\stackrel{\text{(iii)}}{=}\min_{\gamma\in\Gamma}
    \sum_{(i,j)\in\gamma}\Big(D(x_i,y_j)+\log\frac{k}{\rho}\Big)
\\
    &\qquad-\log(1-\rho)(|\mathbb X|-|\gamma|)+\bar\lambda -\left( \log\frac{\bar\lambda}{V}\right )(|\mathbb Y|-|\gamma|)
\\
    &\stackrel{\text{(iv)}}{=}V(1-\rho)+\min_{\gamma\in\Gamma}
    \sum_{(i,j)\in\gamma}\Big(D(x_i,y_j)+\log\frac{k}{\rho}\Big)
\\
    &+\frac{c}{2}(|\mathbb X|+|\mathbb Y|-2|\gamma|)~.
\end{align*}

\cleardoublepage
\balance
\bibliographystyle{IEEEtran}
\bibliography{refs}

\end{document}